\documentclass[a4paper,reqno,twoside]{amsart}

\usepackage[left=3.5cm,right=3.5cm,top=3.5cm,bottom=3.5cm,footskip=1.5cm,headsep=1cm,bindingoffset=0.cm,]{geometry}

\usepackage[utf8]{inputenc}
\usepackage[english]{babel}
\usepackage[T1]{fontenc} 
\usepackage{lmodern} 
\rmfamily 
\DeclareFontShape{T1}{lmr}{b}{sc}{<->ssub*cmr/bx/sc}{}
\DeclareFontShape{T1}{lmr}{bx}{sc}{<->ssub*cmr/bx/sc}{}

\numberwithin{equation}{section}
\usepackage{amsmath}
\usepackage{amssymb}
\usepackage{amsthm}
\usepackage{amsfonts}
\usepackage{mathtools}

\usepackage{mathdots}
\usepackage{caption}
\usepackage{subcaption}

\usepackage{dsfont} 
\usepackage[format=hang]{caption}
\usepackage[normalem]{ulem}

\usepackage{standalone}
\usepackage{graphicx}
\usepackage{imakeidx}
\makeindex
\usepackage{setspace}
\usepackage{appendix}
\usepackage{pdfpages}
\usepackage{upgreek}
\usepackage{tabulary}
\usepackage{booktabs}
\usepackage{extarrows}
\usepackage{tabularx}
\usepackage{float}
\restylefloat{table}
\usepackage{enumitem}
\usepackage{csquotes}
\usepackage{bm}
\usepackage{orcidlink}
\usepackage{lipsum}
\usepackage{nicematrix}

\usepackage{tikz}
\usepackage{tikz-layers}
\usepackage{tikz-cd}
\usepackage[arrow, matrix, curve]{xy}
\usetikzlibrary{matrix,positioning,decorations.pathreplacing,calc}
\usetikzlibrary{
    decorations.text,%
    decorations.markings,%
    shadows}
\usetikzlibrary{calc,intersections}
\usepackage{adjustbox}
\usepackage{graphicx}

\usepackage[
    style=numeric,
    sorting=nty,
    maxnames=99,
    maxalphanames=5,
    natbib=true,
    backend=bibtex,
    sortcites]{biblatex}

\DeclareNameAlias{default}{family-given}

\AtEveryBibitem{
    \clearfield{url}
    \clearfield{issn}
    \clearfield{isbn}
    \clearfield{urldate}

    \ifentrytype{book}{
        \clearfield{pages}}{%
    }
}

\usepackage{xargs}                      
\usepackage[prependcaption,textsize=tiny,textwidth=3cm]{todonotes}
\newcommandx{\unsure}[2][1=]{\todo[linecolor=red,backgroundcolor=red!25,bordercolor=red,#1]{#2}}
\newcommandx{\change}[2][1=]{\todo[linecolor=blue,backgroundcolor=blue!25,bordercolor=blue,#1]{#2}}
\newcommandx{\info}[2][1=]{\todo[linecolor=OliveGreen,backgroundcolor=OliveGreen!25,bordercolor=OliveGreen,#1]{#2}}
\newcommandx{\improvement}[2][1=]{\todo[linecolor=black,backgroundcolor=black!25,bordercolor=black,#1]{#2}}
\newcommandx{\thiswillnotshow}[2][1=]{\todo[disable,#1]{#2}}

\usepackage[noabbrev,capitalize]{cleveref}
\crefname{proposition}{Proposition}{Propositions}
\crefname{equation}{}{}

\newtheorem{theorem}{Theorem}[section]
\newtheorem{lemma}[theorem]{Lemma}

\newtheorem{corollary}[theorem]{Corollary}

\theoremstyle{definition}
\newtheorem{definition}[theorem]{Definition}
\newtheorem{example}[theorem]{Example}

\newtheorem{remark}[theorem]{Remark}

\crefname{assumption}{Assumption}{Assumptions}
\crefname{definition}{Definition}{Definitions}
\crefname{corollary}{Corollary}{Corollaries}
\crefname{enumi}{item}{items}

\creflabelformat{subfigure}{#2\textsc{#1}#3}

\usepackage[final]{microtype}

\makeatletter
\newsavebox\myboxA
\newsavebox\myboxB
\newlength\mylenA

\newcommand*\xoverline[2][0.75]{%
  \sbox{\myboxA}{$\m@th#2$}%
  \setbox\myboxB\null
  \ht\myboxB=\ht\myboxA%
  \dp\myboxB=\dp\myboxA%
  \wd\myboxB=#1\wd\myboxA
  \sbox\myboxB{$\m@th\overline{\copy\myboxB}$}
  \setlength\mylenA{\the\wd\myboxA}
  \addtolength\mylenA{-\the\wd\myboxB}%
  \ifdim\wd\myboxB<\wd\myboxA%
    \rlap{\hskip 0.5\mylenA\usebox\myboxB}{\usebox\myboxA}%
  \else
    \hskip -0.5\mylenA\rlap{\usebox\myboxA}{\hskip 0.5\mylenA\usebox\myboxB}%
  \fi}
\makeatother

\usepackage{fancyhdr}

\fancypagestyle{plain}{%
    \fancyhead[C]{}
    \fancyfoot[C]{}
    
    }
\fancypagestyle{nosection}{%
    \fancyhead[CE]{}
    \fancyhead[CO]{}
    \fancyfoot[CE,CO]{\thepage}
}

\usepackage{tikz}

\usetikzlibrary{matrix} 
\usetikzlibrary{arrows,arrows.meta,bending} 

\usepackage{nicematrix}

\newcommand{\propmat}{P_{v_i,\ell_i,s_i,\gamma_i}(\lambda)}
\newcommand{\symmpropmat}{\Tilde{P}_{v_i,\ell_i,s_i,\gamma_i}(\lambda)}

\usetikzlibrary{hobby}

\usepackage{calc}

\newcommand{\lz}{\ell_2(\Z)}
\newcommand{\ldz}{\mathbb{L}_D(\Z)}
\newcommand{\ldm}{\mathbb{L}_D(M)}

\newcommand{\VJV}{\mc V^\frac{1}{2} \mc J^\gamma \mc V^\frac{1}{2}}

\newcommand{\qq}[1]{\bm q^{(#1)}}
\newcommand{\sss}[1]{\bm s^{(#1)}}
\newcommand{\uu}[1]{\bm u^{(#1)}}

\newcommand{\oo}[1]{\chi^{(#1)}}

\DeclareMathOperator{\len}{len}

\DeclareMathOperator{\Z}{\mathbb{Z}}

\DeclareMathOperator{\R}{\mathbb{R}}
\DeclareMathOperator{\C}{\mathbb{C}}

\DeclareMathOperator{\tr}{tr}
\renewcommand{\i}{\mathbf{i}}

\renewcommand{\qedsymbol}{$\blacksquare$}

\newcommand{\ds}{\displaystyle}

\DeclareMathOperator{\diag}{diag}
\DeclareMathOperator{\BO}{\mathcal{O}}

\DeclareMathOperator{\capmatg}{\mathcal{C}^\gamma}

\renewcommand{\epsilon}{\varepsilon}
\DeclareMathOperator{\dd}{d\!}

\renewcommand{\i}{\mathbf{i}}

\DeclareMathOperator{\iL}{{\mathsf{L}}}
\DeclareMathOperator{\iR}{{\mathsf{R}}}
\DeclareMathOperator{\iLR}{{\mathsf{L},\mathsf{R}}}

\usepackage{tcolorbox}
\usetikzlibrary{tikzmark,calc}

\newcommand{\mc}[1]{\mathcal{#1}}

\newcommand{\abs}[1]{\left\lvert#1\right\rvert}

\newcommand{\norm}[1]{\left\lVert#1\right\rVert}

\newcommand{\veci}[2]{\bm{#1}^{(#2)}}

\newcommandx{\silvio}[2][1=]{\todo[linecolor=blue,backgroundcolor=blue!25,bordercolor=blue,#1]{Silvio: #2}}

\newcommandx{\alex}[2][1=]{\todo[linecolor=red,backgroundcolor=red!25,bordercolor=red,#1]{Alex: #2}}

\newcommandx{\clemens}[2][1=]{\todo[linecolor=cyan,backgroundcolor=cyan!25,bordercolor=cyan,#1]{Clemens: #2}}

\title[Competing edge and bulk localisation in non-reciprocal disordered systems]{Competing edge and bulk localisation in non-reciprocal disordered systems}

\begin{document}
\author[H. Ammari]{Habib Ammari \,\orcidlink{0000-0001-7278-4877}}
\address{\parbox{\linewidth}{Habib Ammari\\
 ETH Z\"urich, Department of Mathematics, Rämistrasse 101, 8092 Z\"urich, Switzerland, \href{http://orcid.org/0000-0001-7278-4877}{orcid.org/0000-0001-7278-4877}}.}
 \email{habib.ammari@math.ethz.ch}
 \thanks{}

\author[S. Barandun]{Silvio Barandun\,\orcidlink{0000-0003-1499-4352}}
  \address{\parbox{\linewidth}{Silvio Barandun\\
 ETH Z\"urich, Department of Mathematics, Rämistrasse 101, 8092 Z\"urich, Switzerland, \href{http://orcid.org/0000-0003-1499-4352}{orcid.org/0000-0003-1499-4352}}.}
 \email{silvio.barandun@sam.math.ethz.ch}

\author[C. Thalhammer]{Clemens Thalhammer}
\address{\parbox{\linewidth}{Clemens Thalhammer\\
 ETH Z\"urich, Department of Mathematics, Rämistrasse 101, 8092 Z\"urich, Switzerland.}}
 \email{clemens.thalhammer@sam.math.ethz.ch}

 \author[A. Uhlmann]{Alexander Uhlmann\,\orcidlink{0009-0002-0426-6407}}
  \address{\parbox{\linewidth}{Alexander Uhlmann\\
  ETH Z\"urich, Department of Mathematics, Rämistrasse 101, 8092 Z\"urich, Switzerland, \href{http://orcid.org/0009-0002-0426-6407}{orcid.org/0009-0002-0426-6407}}.}
 \email{alexander.uhlmann@sam.math.ethz.ch}

\begin{abstract}
We investigate the competing mechanisms of localisation in one-dimensional block disordered subwavelength resonator systems subject to non-reciprocal damping, induced by an imaginary gauge potential. Using a symmetrisation approach to enable the adaptation of tools from Hermitian systems, we derive the limiting spectral distribution of these systems as the number of blocks goes to infinity and characterise their spectral properties in terms of the spectral properties of their constituent blocks. By employing a transfer matrix approach, we then clarify, in terms of Lyapunov exponents, the competition between the edge localisation due to imaginary gauge potentials and the bulk localisation due to disorder. In particular, we demonstrate how the disorder acts as insulation against the non-Hermitian skin effect, preventing edge localisation for small imaginary gauge potentials.
\end{abstract}

\maketitle

\date{}

\bigskip

\par

\bigskip

\\

\section{Introduction}
In this paper, we study the spectral properties of large one-dimensional resonator systems with non-reciprocal, directional damping terms. Due to non-reciprocity, wave propagation in such systems is usually amplified in one direction and attenuated in the other. This inherent unidirectional dynamics is related to the non-Hermitian skin effect, which leads to the accumulation of modes at one edge of the structure, a phenomenon unique to non-Hermitian systems with non-reciprocal coupling; see \cite{yao.wang2018Edge} where the term “non-Hermitian skin effect” was first coined. 
We also refer the reader to \cite{hatano.nelson1996Localization,yokomizo.yoda.ea2022NonHermitiana,rivero.feng.ea2022Imaginary, zhang.zhang.ea2022review, okuma.kawabata.ea2020Topological,lin.tai.ea2023Topological,yokomizo.yoda.ea2022NonHermitiana,wang.chong2023NonHermitian,leykam.bliokh.ea2017Edge,borgnia.kruchkov.ea2020Nonhermitian,okuma.sato2023Nonhermitian,ammariMathematicalFoundationsNonHermitian2024,ammari2023skin3d} for the theoretical understanding of the skin effect in non-Hermitian physical systems and to \cite{ghatak.brandenbourger.ea2020Observation,longhi.gatti.ea2015Robust,franca.konye.ea2022Nonhermitian, wang.wang.ea2022NonHermitian} for its experimental realisation in topological photonics, phononics, and other condensed matter systems.

Wave propagation in the systems of resonators considered in this paper is modelled by the Helmholtz equation \eqref{eq:coupled ods} modified by an imaginary gauge potential $\gamma$, where we assume that the directional damping terms are added only inside the subwavelength resonators and that the contrast between the material parameters inside and outside the resonators, $\rho_b/\rho$, is small and the wave speeds $v$ and $v_i$ for $i=1,\ldots, N,$ are of order one with $N$ being the number of resonators. As shown in \cite{ammariMathematicalFoundationsNonHermitian2024}, the directional damping terms break the Hermiticity, time-reversal symmetry, and reciprocity of the problem and constitute the crucial mechanism responsible for the condensation effects at one edge of the system. It is worth emphasising that their effect would be negligible without exciting the system's subwavelength resonant frequencies; see, for instance, \cite{ammari.davies.ea2024Functional}.

For a periodically arranged resonator system, based on Toeplitz theory, it was proved in \cite{ammariMathematicalFoundationsNonHermitian2024} that the system possesses $N$ subwavelength resonant frequencies, \emph{i.e.}, with corresponding eigenmodes having wavelengths much larger than the typical size of the resonators or, equivalently, being almost constant inside the resonators and accumulating, depending on the sign of $\gamma$, at one edge of the structure. Furthermore, the subwavelength resonant frequencies and eigenmodes can be approximated using the generalised gauge capacitance matrix given by $VC^\gamma$, where $C^\gamma$ and $V$ are defined by \cref{eq:Cdef} and \cref{eq:vgammadef}, respectively. On the other hand, as proved in \cite{stabilityskin}, the skin effect in systems of non-Hermitian subwavelength resonators is quite robust with respect to small random perturbations in the system's material parameters. 
Nevertheless, for large enough random perturbations, some eigenmodes become localised in the bulk. Moreover, as the strength of the random perturbations increases, more and more eigenmodes become localised in the bulk. This competition between the localisation at one edge of the structure and the localisation in the bulk was first illustrated numerically in \cite{stabilityskin}. 

In this work,  we consider large arrays of subwavelength resonators with non-reciprocal, directional damping terms, constructed by sampling randomly from a fixed set of \enquote{blocks} and analyse their unusual spectral behaviour.  

By symmetrising the generalised gauge capacitance matrix, we derive the limiting spectral distribution of non-reciprocal block disordered systems as the number of blocks goes to infinity and characterise their spectral gaps in terms of the spectral gaps of its constituent blocks. We also elucidate the mechanisms of localisation in disordered non-reciprocal systems. We show that the non-reciprocity pushes the eigenmodes to be localised at the edge as a result of the imaginary gauge potentials, while the disorder causes the eigenmodes to become localised in the bulk. Exhibiting these two mechanisms simultaneously leads to an intricate competition between edge and bulk localisation.
To understand this localisation, we investigate the Lyapunov exponent, which, as shown in \cite{ammari.barandun.ea2025Subwavelength}, allows us to predict the decay of a given eigenmode at a given frequency. Finally, we show how the disorder acts as an \emph{insulation} against the non-Hermitian skin effect, preventing edge localisation for small imaginary gauge potentials, while for finite periodic systems consisting of repeated identical unit cells, any small gauge potential leads to complete localisation of all eigenmodes at the edge.

This paper is organised as follows. In \cref{sec:setting}, we formulate the setting of non-reciprocal block disordered arrays of subwavelength resonators and introduce the fundamental tools to study them, namely the (symmetrised) capacitance matrix and propagation matrices. In \cref{sec:jacobi}, we adapt the tools introduced in \cite{ammari.barandun.ea2025Universal} to the symmetrised capacitance matrix to investigate the spectral convergence as the number of blocks $M\to \infty$. In \cref{sec:dos}, we employ the propagation matrix formalism to more explicitly characterise the density of states of the resonator system, based on the properties of its constituent blocks. Finally, in \cref{sec:competition}, we give a decomposition of the system's Lyapunov exponent into a disordered term and a non-reciprocal term. This enables a precise discussion of the competition between disorder-induced bulk localisation and non-reciprocity-induced edge localisation.

\section{Setting}\label{sec:setting}
In this section, we first introduce the non-reciprocal subwavelength setting and recall from \cite{ammariMathematicalFoundationsNonHermitian2024} the discrete approximation of the eigenfrequencies and eigenmodes of the resonator system in terms of the eigenvalues and eigenvectors of the gauge capacitance matrix, which can be symmetrised. Then, we introduce the notation and the construction associated with block disordered non-reciprocal systems and adapt the propagation matrix approach from  \cite{ammari.barandun.ea2025Subwavelength} to the non-reciprocal setting. 

\subsection{non-reciprocal arrays of subwavelength resonators}
As in \cite{ammari.barandun.ea2025Universal,ammari.barandun.ea2025Subwavelength}, the central systems of interest are one-dimensional chains of $N$ disjoint subwavelength resonators $D_i\coloneqq (x_i^{\iL},x_i^{\iR})$, where $(x_i^{\iLR})_{1\<i\<N} \subset \R$ are the $2N$ extremities satisfying $x_i^{\iL} < x_i^{\iR} <  x_{i+1}^{\iL}$ for any $1\leq i \leq N$. We fix the coordinates such that $x_1^{\iL}=0$. We denote by $\ell_i = x_i^{\iR} - x_i^{\iL}$ the length of the $i$\textsuperscript{th} resonator,  and by $s_i= x_{i+1}^{\iL} -x_i^{\iR}$ the spacing between the $i$\textsuperscript{th} and $(i+1)$\textsuperscript{th} resonator. We use 
\begin{align*}
  \mc D\coloneqq \bigcup_{i=1}^N(x_i^{\iL},x_i^{\iR})
\end{align*}
to symbolise the set of subwavelength resonators.

In this work, following \cite{ammariMathematicalFoundationsNonHermitian2024,yokomizo.yoda.ea2022NonHermitiana}, we consider the following generalised Sturm--Liouville equation where an imaginary gauge potential $\gamma$ is introduced:
\begin{align}
    -\frac{\omega^{2}}{\kappa(x)}u(x)- \gamma(x) \frac{\dd}{\dd x}u(x)-\frac{\dd}{\dd x}\left( \frac{1}{\rho(x)}\frac{\dd}{\dd
    x}  u(x)\right) =0,\qquad x \in\R,
    \label{eq: gen Sturm-Liouville}
\end{align}
for piecewise constant coefficients
\begin{align*}
\gamma(x) = \begin{dcases}
    \gamma_i,\quad x\in D_i,\\
    0,\quad x \in \R\setminus \mc D,
\end{dcases}
\end{align*}
\begin{align*}
    \kappa(x)=
    \begin{dcases}
        \kappa_i & x\in D_i,\\
        \kappa&  x\in\R\setminus \mc D,
    \end{dcases}\quad\text{and}\quad
    \rho(x)=
    \begin{dcases}
        \rho_b & x\in D_i,\\
        \rho&  x\in\R\setminus \mc D,
    \end{dcases}
\end{align*}
where the constants $\gamma_i, \rho_b, \rho, \kappa_i, \kappa \in \R_{>0}$. The wave speeds inside the resonators $D_i$ and inside the background medium $\R\setminus \mc D$, are denoted respectively by $v_i$ and $v$,  and the contrast between the densities of the resonators and the background medium by $\delta$:
\begin{align}
    v_i:=\sqrt{\frac{\kappa_i}{\rho_b}}, \qquad v:=\sqrt{\frac{\kappa}{\rho}}, \qquad
    \delta:=\frac{\rho_b}{\rho}.
\end{align}

In these circumstances of piecewise constant material parameters, the wave problem determined by \eqref{eq: gen Sturm-Liouville} can be rewritten as the following system of coupled one-dimensional equations:

\begin{align}
    \begin{dcases}
    \frac{{\dd}^{2}}{\dd x^2}u(x) + \gamma_i \frac{\dd}{\dd x}u(x) +\frac{\omega^2}{v_i^2}u(x)=0, & x\in D_i,\\
    \frac{\dd^2}{\dd x^2}u(x)  + \frac{\omega^2}{v^2}u(x)=0, & x\in\R\setminus \mc D,\\
        u\vert_{\iR}(x^{\iLR}_i) - u\vert_{\iL}(x^{\iLR}_i) = 0, & \text{for all } 1\leq i\leq N,\\
        \left.\frac{\dd u}{\dd x}\right\vert_{\iR}(x^{\iL}_{{i}})=\delta\left.\frac{\dd u}{\dd x}\right\vert_{\iL}(x^{\iL}_{{i}}), & \text{for all } 1\leq i\leq N,\\
        \delta\left.\frac{\dd u}{\dd x}\right\vert_{\iR}(x^{\iR}_{{i}})=\left.\frac{\dd u}{\dd x}\right\vert_{\iL}(x^{\iR}_{{i}}), & \text{for all } 1\leq i\leq N,\\
        \frac{\dd u}{\dd\ \abs{x}}(x) -\i \frac{\omega}{v} u(x) = 0, & x\in(-\infty,x_1^{\iL})\cup (x_N^{\iR},\infty),
    \end{dcases}
\label{eq:coupled ods}
\end{align}
where, for a one-dimensional function $w$, we denote by
\begin{align*}
    w\vert_{\iL}(x) \coloneqq \lim_{\substack{s\to 0\\ s>0}}w(x-s) \quad \mbox{and} \quad  w\vert_{\iR}(x) \coloneqq \lim_{\substack{s\to 0\\ s>0}}w(x+s)
\end{align*}
if the limits exist. 

We are interested in resonances $\omega\in\C$ such that \eqref{eq:coupled ods} has a nontrivial solution. In \cite{ammariMathematicalFoundationsNonHermitian2024}, an asymptotic analysis is derived in a high contrast limit, given by $\delta\to 0$. 
Within this regime, the subwavelength resonant frequencies and their corresponding eigenmodes can be characterised as the eigenvalues and eigenvectors of the so-called \emph{gauge capacitance matrix}. The following results are from \cite{ammariMathematicalFoundationsNonHermitian2024}.
\begin{theorem} \label{thm:gauge}
    Let $\gamma \neq 0$. Let the gauge capacitance matrix $C^\gamma =(C_{i,j}^\gamma)^N_{i,j=1}$ and the \emph{material matrix} $V$ be respectively
 defined by
    \begin{align}
        \label{eq:Cdef}
        C_{i,j}^\gamma \coloneqq \begin{dcases}
            \frac{\zeta(\gamma_1\ell_1)}{s_1},                    & i=j=1,                 \\
             \frac{\zeta(-\gamma_i\ell_i)}{s_{i-1}}+\frac{\zeta(\gamma_i\ell_i)}{s_i},                & 1< i=j< N,             \\
            \frac{\zeta(-\gamma_N\ell_N)}{ s_{N-1}},       & i=j=N,                 \\
            -\frac{\zeta(\pm\gamma_i\ell_i)}{s_{\min(i,j)}} , & 1\leq i\pm 1=j \leq N, \\
            \ 0,                                                 & \text{else},
        \end{dcases}
    \end{align}
where \begin{equation} \label{def:zeta}
\zeta(z)\coloneqq \frac{z}{1-e^{-z}}>0
\end{equation} 
and
\begin{gather} 
    V = \diag \left({\frac{v_1^2}{\ell_1}}, \dots,  \frac{v_N^2}{\ell_N} \right) \in \mathbb{R}^{N\times N}. \label{eq:vgammadef}
 \end{gather} 
    Then, the $N$ subwavelength eigenfrequencies $\omega_i$ of (\ref{eq: gen Sturm-Liouville}) associated to this system satisfy, as $\delta\to0$,
              \begin{align*}
                  \omega_i =  \sqrt{\delta\lambda_i} + \BO(\delta),
              \end{align*}
              where $(\lambda_i)_{1\leq i\leq N}$ are the eigenvalues of the eigenvalue problem 
              \begin{equation} \label{eq:ch}
              V C^\gamma \bm u_i = \lambda_i \bm u_i.
              \end{equation}
              Furthermore, let $u_i$ be a subwavelength eigenmode corresponding to $\omega_i$ and let $\bm u_i$ be the corresponding eigenvector of the generalised gauge capacitance matrix $\capmatg:= VC^\gamma$. Then, 
              \begin{equation} \label{eq:ui}
                  u_i(x) = \sum_j \bm u_i^{(j)}V_j(x) + \BO(\delta),
              \end{equation}
              where $\bm u_i^{(j)}$ is the $j$\textsuperscript{th} entry of the vector $\bm u_i$ and $V_j(x)$ are defined as the solution to
              \begin{align}
                  \begin{dcases}
                      -\frac{\dd{^2}}{\dd x^2} V_j(x) =0, & x\in\R\setminus \mc D, \\
                      V_j(x)=\delta_{ij},              & x\in D_i,                          \\
                      V_j(x) = \BO(1),                  & \mbox{as } \abs{x}\to\infty.
                  \end{dcases}
                  \label{eq: def V_i}
              \end{align}
\end{theorem}

\subsection{Gauge capacitance matrix symmetrisation}
From \cref{eq:Cdef}, we can see that the non-reciprocal nature of our system causes the gauge capacitance matrix $C^\gamma$ to no longer be symmetric, as opposed to the Hermitian case. However, in one dimension, we can exploit the tridiagonal nature of $C^\gamma$ to find a similar but symmetric system by factoring out the non-reciprocal decay.

\begin{lemma}
    Define the \emph{symmetrised gauge capacitance matrix} $J^\gamma$ as 
    \begin{gather}
    J^\gamma = \left(\begin{array}{cccccc}
         \bm q^{(1)}& -\bm s^{(1)}\\
         -\bm s^{(1)}& \bm q^{(2)}& -\bm s^{(2)} \\
         &\ddots&\ddots&\ddots \\
         &&-\bm s^{(N-2)}& \bm q^{(N-1)}& -\bm s^{(N-1)}\\
         &&&-\bm s^{(N-1)}& \bm q^{(N)}&
    \end{array}\right) \in \mathbb{R}^{N\times N} \label{eq:symmcapmat}
    \end{gather}
    with diagonal and off-diagonal given by
    \begin{equation}\label{eq:symmcapmatbands}
        \bm q^{(i)} \coloneqq C^\gamma_{i,i}, \quad \bm s^{(i)} \coloneqq \sqrt{\frac{\zeta(\gamma_i\ell_i)\zeta(-\gamma_{i+1}\ell_{i+1})}{s_i^2}}.
    \end{equation}
    Furthermore,  introduce the \emph{transformation matrix} $R = \diag\left(\bm r^{(1)},\dots,\bm r^{(N)}\right)$ with $r_i$ defined recursively as
    \begin{equation}\label{eq:recursiveR}
        \veci{r}{1} = \sqrt{\zeta(-\gamma_1\ell_1)}, \quad \veci{r}{i+1} = \left(\sqrt{\frac{\zeta(-\gamma_{i+1}\ell_{i+1})}{\zeta(\gamma_{i}\ell_{i})}}\right)\veci{r}{i}.
    \end{equation}

    Then $C^\gamma$ is similar to the symmetric matrix $J^\gamma$ via the transform $R$, namely
    \begin{equation}\label{eq:tridiagsimilarity}
        J^\gamma = R^{-1} C^\gamma R .
    \end{equation}
\end{lemma}
The proof of this lemma lies in a direct verification of \cref{eq:tridiagsimilarity}.

As a consequence, we immediately have
\begin{equation}\label{eq:sameeigenvalues}
    VJ^\gamma \bm v_i = \lambda_i\bm v_i \iff VC^\gamma \bm u_i = \lambda_i\bm u_i \quad \text{ with }\bm u_i = R\bm v_i,
\end{equation}
where we used \cref{eq:tridiagsimilarity} and the fact that both $R^{-1}$ and $V$ are diagonal matrices and thus commute.
Hence, any eigenpair $(\lambda_i,\bm v_i)$ of $VJ^\gamma$ corresponds to an eigenpair $(\lambda_i,\bm u_i = R\bm v_i)$ of $\mc C^\gamma = VC^\gamma$. 

\begin{corollary}
    The generalised gauge capacitance matrix $\mc C^\gamma = VC^\gamma$ has $N$ real, distinct and non-negative eigenvalues.
\end{corollary}
\begin{proof}
    Due to \cref{eq:sameeigenvalues} we may equivalently investigate the eigenvalues of $VJ^\gamma\sim V^\frac{1}{2}J^\gamma V^{\frac{1}{2}}$. Now, $V^\frac{1}{2}J^\gamma V^{\frac{1}{2}}$ is a real symmetric, tridiagonal matrix and thus has $N$ real, distinct eigenvalues. 

    Finally, from \cref{eq:Cdef} it follows that 
    \[
        (VC^\gamma)_{ii} = \sum_{j\neq i}\abs{(VC^\gamma)_{ij}}
    \]
    which by the Gershgorin circle theorem implies that all eigenvalues of $VC^\gamma$ lie in the non-negative half plane, as desired.
\end{proof}

\begin{remark}
    The transformation matrix $R$ encodes the decay due to non-reciprocity of the eigenmodes of $\mc C^\gamma$. 
    We can explicitly identify this decay as 
    \begin{equation}\label{eq:rdecay}
        \veci{r}{i+1} = \sqrt{\zeta(-\gamma_{i+1}\ell_{i+1})}e^{-\frac{1}{2}\sum_{k=1}^i\gamma_k\ell_k},
    \end{equation}
    by expanding the recursive definition \cref{eq:recursiveR} and using the fact that $\frac{\zeta(-z)}{\zeta(z)} = e^{-z}$. 

    The exponential part of \cref{eq:rdecay} quickly dominates, causing $\veci{r}{i}$ to decay exponentially as $i$ increases. As a consequence, any eigenmode $\bm v_i$ of $V J^\gamma$ yields an eigenmode $\bm u_i = R\bm v_i$ of $\mc C^\gamma$ with decay $\bm u_i^{(j)} \approx \bm v_i^{(j)}e^{-\frac{1}{2}\sum_{k=1}^j\gamma_k\ell_k}$.
\end{remark}

\subsection{Non-reciprocal block disordered systems}
In this work, we are interested in the localisation properties of disordered non-reciprocal systems, where local translation invariance is broken. To that end, we adapt the \emph{block disordered formalism} from \cite{ammari.barandun.ea2025Universal,ammari.barandun.ea2025Subwavelength} to the non-reciprocal setting. Namely, we consider a finite number of distinct \enquote{building blocks} consisting of (possibly multiple) subwavelength resonators. As in the Hermitian setting considered in \cite{ammari.barandun.ea2025Universal,ammari.barandun.ea2025Subwavelength}, constructing disordered resonator arrays from simple building blocks with imaginary gauge potentials will enable us to characterise the ``unusual'' localisation properties of the array by looking at the building blocks.

We recall that a subwavelength block disordered system is a chain of subwavelength resonators consisting of $M$ blocks of resonators $B_{\oo{j}}$ sampled according to a sequence $$\chi \in \{1,\dots , D\}^M \eqqcolon \ldm$$ from a selection $B_1, \dots, B_D$ of $D$ distinct resonator blocks, arranged along a line. Each resonator block $B_j$ is characterised by a sequence of tuples 
\[
(v_1,\ell_1,s_1,\gamma_1), \dots, (v_{\len(B_j)},\ell_{\len(B_j)},s_{\len(B_j)}, \gamma_{\len(B_j)}), 
\] which denotes the wave speed, length, spacing, and non-reciprocal decay of each constituent resonator. Here, $\len(B_j)$ denotes the total number of resonators contained within the block $B_j$.

We will often abuse notation and write $\mc D = \bigcup_{j=1}^M B_{\oo{j}}$ to denote the resonator array constructed by sampling the blocks $B_1, \dots B_D$  according to $\chi\in \ldm$ and then arranging them in a line. Having thus constructed an array of subwavelength resonators, we can alternatively write $\mc D = \bigcup_{i=1}^N D_i$.
Note that since $M$ denotes the total number of sampled blocks and $N$ the total number of resonators, we always have $M\leq N$.

\begin{figure}
    \centering
    \begin{tikzpicture}[scale=1.0, every node/.style={font=\footnotesize}]

    \draw[thick] (0,0) -- (2,0);
    \node[above] at (1,0) {$D_1$};
    \draw[->] (2,-0.2) -- node[below] {$\gamma_1(B_1)$} (0,-0.2);
    
    \draw[
      decorate,
      decoration={brace, mirror, amplitude=5pt}
    ] (0,-0.6) -- (4,-0.6)
      node[midway, yshift=-0.45cm] {$B_1$};
    
    \draw[-, dotted] (0,0) -- (0,0.7);
    \draw[-, dotted] (2,0) -- (2,0.7);
    \draw[|-|, dashed] (0,0.7) -- node[above]{$\ell_1(B_1)$} (2,0.7);
    
    \draw[-, dotted] (4,0) -- (4,0.7);
    \draw[|-|, dashed] (2,0.7) -- node[above]{$s_1(B_1)$} (4,0.7);
    
    \draw[thick] (4,0) -- (5,0);
    \node[above] at (4.5,0) {$D_2$};
    \draw[->] (5,-0.2) -- node[below] {$\gamma_1(B_2)$} (4,-0.2);
    
    \draw[thick] (6,0) -- (7,0);
    \node[above] at (6.5,0) {$D_3$};
    \draw[->] (7,-0.2) -- node[below] {$\gamma_2(B_2)$} (6,-0.2);
    
    \draw[
      decorate,
      decoration={brace, mirror, amplitude=5pt}
    ] (4,-0.6) -- (9,-0.6)
      node[midway, yshift=-0.45cm] {$B_2$};
    
    \draw[-, dotted] (5,0) -- (5,0.7);
    \draw[-, dotted] (6,0) -- (6,0.7);
    \draw[-, dotted] (7,0) -- (7,0.7);
    \draw[-, dotted] (9,0) -- (9,0.7);
    \draw[|-|, dashed] (4,0.7) -- node[above]{$\ell_1(B_2)$} (5,0.7);
    \draw[|-|, dashed] (5,0.7) -- node[above]{$s_1(B_2)$} (6,0.7);
    \draw[|-|, dashed] (6,0.7) -- node[above]{$\ell_2(B_2)$} (7,0.7);
    \draw[|-|, dashed] (7,0.7) -- node[above]{$s_2(B_2)$} (9,0.7);
    
    \begin{scope}[shift={(9,0)}]
        \draw[thick] (0,0) -- (2,0);
        \node[above] at (1,0) {$D_4$};
        \draw[->] (2,-0.2) -- node[below] {$\gamma_1(B_1)$} (0,-0.2);
        
        \draw[
          decorate,
          decoration={brace, mirror, amplitude=5pt}
        ] (0,-0.6) -- (4,-0.6)
          node[midway, yshift=-0.45cm] {$B_1$};
        
        \draw[-, dotted] (0,0) -- (0,0.7);
        \draw[-, dotted] (2,0) -- (2,0.7);
        \draw[|-|, dashed] (0,0.7) -- node[above]{$\ell_1(B_1)$} (2,0.7);
        
        \draw[-, dotted] (4,0) -- (4,0.7);
        \draw[|-|, dashed] (2,0.7) -- node[above]{$s_1(B_1)$} (4,0.7);
    \end{scope}
    
\end{tikzpicture}
    \caption{A non-reciprocal block disordered system consisting of two monomer blocks $B_1$ and a dimer block $B_2$ arranged in a chain given by the sequence $\chi= (1,2,1)$. It thus consists of $M=3$ blocks and $N=4$ resonators $D_1,\dots ,D_4$ in total.}
    \label{fig:disorderedsketch}
\end{figure}
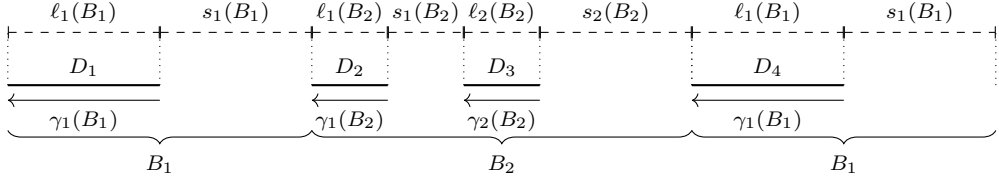
\begin{example}\label{ex:standard_blocks}
    The central example in this work are block disordered systems obtained by sampling from a two block set $B_1$ and $B_2$ ($D=2$), a monomer (\emph{i.e.} single resonator) $B_1=B_{monomer}$ and dimer block $B_2=B_{dimer}$.
    In particular, we construct the blocks as follows: For the monomer block $\len(B_{monomer}) = 1$ and $\ell_1(B_{monomer}) = s_1(B_{monomer})  = 2$ while for the dimer block we have $\len(B_{dimer}) = 2$ and $\ell_1(B_{dimer}) = \ell_2(B_{dimer})=1$ and $s_1(B_{dimer})=1, s_2(B_{dimer})=2$. We choose all wave speeds to be equal to $1$. Furthermore, we choose all imaginary gauge potentials to be equal $\gamma_1(B_{monomer}) = \gamma_1(B_{dimer}) = \gamma_2(B_{dimer}) = \gamma$. This has the effect that $\ell_1(B_{monomer})\gamma_1(B_{monomer}) = \ell_1(B_2)\gamma_1(B_2) +\ell_2(B_{dimer})\gamma_2(B_{dimer})$ thus ensuring that the decay due to non-reciprocity is equal across the monomer and dimer blocks. 
    An example of a single realisation of this system with $M=14$ is described by the following sequence:
    \[
        \chi = (1,1,1,1,2,1,1,1,1,1,2,1,1,1) \in \{1,2\}^M . 
    \]
    In \cref{fig:disorderedsketch}, another realisation corresponding to the sequence $\chi = (1,2,1)$ is illustrated.
\end{example}

\subsection{Transfer and propagation matrices}\label{ssec:transfer and propmat}
We now aim to adapt the powerful propagation matrix approach from \cite{ammari.barandun.ea2025Subwavelength} to the non-reciprocal setting. To that end, we consider an arbitrary block disordered system $\mc D = \bigcup_{j=1}^M B_{\oo{j}} = \bigcup_{i=1}^ND_i$ consisting of blocks $B_1, \dots B_D$ arranged in a sequence $\chi\in \ldm$ and described by the generalised gauge capacitance matrix $\mc C^\gamma = VC^\gamma$.
By repeating the arguments in \cite[Section 4]{ammari.barandun.ea2025Subwavelength} for $\mc C^\gamma$, we find the non-reciprocal propagation matrix governing the subwavelength solutions of \cref{eq:coupled ods}.

\begin{definition}[Non-reciprocal propagation matrices]
    For the $i$\textsuperscript{th} subwavelength resonator with wave speed $v_i$, length $\ell_i$, spacing $s_i$ and decay $\gamma_i$, we define the \emph{propagation matrix} at frequency $\lambda$ to be the $2$-by-$2$ matrix that propagates any solution $u(x)$ of the subwavelength resonance problem \cref{eq:coupled ods} from the left edge of the resonator $x_i^{\iL}$ to the left edge of the following resonator $x_{i+1}^{\iL}$. Namely
    \begin{equation}\label{eq:propmatdef}
    \begin{pmatrix}
        u_-(x^{\iL}_{i+1})\\
        u_-'(x^{\iL}_{i+1})
    \end{pmatrix} = \underbrace{\begin{pmatrix}
        1-s_i\frac{\ell_i}{\zeta(\gamma_i\ell_i)v_i^2}\lambda & e^{-\gamma_i\ell_i}s_i\\
        -\frac{\ell_i}{\zeta(\gamma_i\ell_i)v_i^2}\lambda & e^{-\gamma_i\ell_i}
    \end{pmatrix}}_{\propmat\coloneqq} \begin{pmatrix}
        u_-(x^{\iL}_{i})\\
        u'_-(x^{\iL}_{i})
    \end{pmatrix},
\end{equation}
where $u_-(x^{\iL}_{i}) = \lim_{x\uparrow x^{\iL}_i}u(x)$ and $u'_-(x^{\iL}_{i})= \lim_{x\uparrow x^{\iL}_i}u'(x)$ denote the \emph{exterior} values of the eigenmode $u(x)$. We will denote $P_i(\lambda) = \propmat$.

Due to the non-Hermitian nature of our system, we find that 
$$\det \propmat = e^{-\gamma_i\ell_i} \neq 1,$$
where $\det$ denotes the determinant. 

It will often prove useful to factor out the non-reciprocal decay and define the \emph{symmetrised propagation matrix} 
\begin{equation}
    \symmpropmat \coloneqq e^{\frac{1}{2}\gamma_il_i}\propmat, 
\end{equation} satisfying $\det \symmpropmat = 1$.
\end{definition}

It turns out that, using a suitable change of basis, the \emph{symmetrised propagation matrix} corresponds exactly to the \emph{symmetrised gauge capacitance matrix} $VJ^\gamma$. 
Namely, using \cref{eq:symmcapmat}, we find its \emph{symmetrised transfer matrix} is given by
\begin{equation}
    \Tilde{T}_i(\lambda) \coloneqq \begin{pmatrix}
        \frac{\veci{q}{i}-\ell_iv_i^{-2}\lambda}{\veci{s}{i}} & -\frac{\veci{s}{i-1}}{\veci{s}{i}}\\
        1 & 0
    \end{pmatrix}
\end{equation}
for $1\leq i\leq N$ with $\veci{q}{i}$ and $\veci{s}{i}$ as in \cref{eq:symmcapmatbands}. To define $\Tilde{T}_1(\lambda)$ and $\Tilde{T}_N(\lambda)$, we choose the boundary values corresponding to extending the array periodically, namely $(v_0, l_0, \gamma_0) \coloneqq (v_N, l_N, \gamma_N)$, $(v_{N+1}, l_{N+1}, \gamma_{N+1}) \coloneqq (v_1, l_1, \gamma_1)$ and $s_0 = s_N \coloneqq s_{\len (B_{\oo{M}})} (B_{\oo{M}})$. However, this is merely a matter of convention, as any consistent choice of boundary values ultimately yields the same results.  
$\Tilde{T}_i(\lambda)$ is called the transfer matrix since any eigenvector $\bm v \in \C^N$ satisfying $(VJ^\gamma-\lambda I_N)\bm v= 0$, must satisfy
\begin{equation}
    \Tilde{T}_i(\lambda)\begin{pmatrix}
        \veci{v}{i}\\\veci{v}{i-1}
    \end{pmatrix} = \begin{pmatrix}
        \veci{v}{i+1}\\ \veci{v}{i}
    \end{pmatrix},
\end{equation}
for any $2\leq i\leq N-1$.

We can now define the change of basis 
\begin{equation}
    Q_i \coloneqq \begin{pmatrix}
        \sqrt{\zeta(-\gamma_{i+1}\ell_{i+1})} & 0\\
        \frac{\sqrt{\zeta(-\gamma_{i+1}\ell_{i+1})}}{s_i} & -e^{\frac{\gamma_i\ell_i}{2}}\frac{\sqrt{\zeta(-\gamma_{i}\ell_{i})}}{s_i}
    \end{pmatrix}
\end{equation}
and find 
\begin{equation}\label{eq:propmat-transfermat similarity}
    \Tilde{P}_i(\lambda) = Q_i\Tilde{T}_i(\lambda)Q_{i-1}^{-1}
\end{equation}
for any $1\leq i \leq N$. Note that this change of basis corresponds to a combination of switching from symmetrised coordinates back to regular coordinates and then making the same change of basis from $(\veci{u}{i}, \veci{u}{i-1})^\top$ to $(u_-(x_i^{\iL}), u_-'(x_i^{\iL}))^\top$ as in \cite{ammari.barandun.ea2025Subwavelength}. Here, the superscript $^\top$ denotes the transpose. 

\section{Convergence of spectral distributions}\label{sec:jacobi}
In this section, our aim is to understand the limiting spectral distribution of 
non-reciprocal block disordered systems as the number of blocks $M$ goes to infinity. We shall achieve this by adapting the arguments from \cite[Section 3]{ammari.barandun.ea2025Universal} to the non-reciprocal case. Namely, by studying the limiting spectral distribution of the \emph{symmetrised} gauge capacitance matrix, we are able to transfer many tools from the Hermitian setting to the non-reciprocal one.

\subsection{Preliminary results}
Given an arbitrary symmetric matrix $A\in \R^{N\times N}$, we define the \emph{empirical cumulative density of states} to be the number of eigenvalues $\lambda_i \in \sigma(A)$ of $A$ that do not exceed $\lambda$, that is, 
\begin{equation}
    D(A,\lambda) \coloneqq \frac{1}{N}\abs{\sigma(A) \cap (-\infty,\lambda]}
\end{equation}
where $\abs{S}$ denotes the \emph{cardinality} of a finite set $S$.

The central result in this section will be that for any independent and identical distribution of blocks, $D(\mc C^\gamma,\lambda)$ converges weakly to a non-random limiting distribution as the number of blocks sampled $M$ goes to infinity. 

To establish this result, we study the limiting infinite non-reciprocal block disordered system obtained by sampling blocks according to an infinite sequence $\chi\in \ldz$ which is independent and identically distributed with block probabilities $p_1,\dots,p_D$. We thus obtain an infinite array of blocks $\mc D = \bigcup_{j=-\infty}^\infty B_{\oo{j}} = \bigcup_{i=-\infty}^\infty D_i$, which in turn uniquely determines the bi-infinite sequences of wave speeds $(v_i)_{i=-\infty}^\infty$, resonator lengths $(\ell_i)_{i=-\infty}^\infty$, spacings $(s_i)_{i=-\infty}^\infty$ and imaginary gauge potentials $(\gamma_i)_{i=-\infty}^\infty$.

In this limit, the symmetrised gauge capacitance matrix $J^\gamma$ turns into the \emph{symmetrised Jacobi operator} given by
\begin{align*}
        \mc J^\gamma:\lz &\to \lz\\
        \uu{i} &\mapsto (J\bm u)^{(i)} = -\sss{i-1}\uu{i-1} + \qq{i}\uu{i} - \sss{i}\uu{i+1}
\end{align*} 
with $\sss{i}$ and $\qq{i}$ as in \cref{eq:symmcapmatbands}
and the material matrix $V$ into the \emph{material operator} 
\begin{equation}
    \mc V = \diag \left(\dots, \frac{v_i^2}{\ell_i}, \dots\right) :\lz \to \lz.
\end{equation}

Since we aim to employ the tools developed for Jacobi operators, which must be symmetric, to study $\mc V \mc J^\gamma$ we exploit the fact that $\mc V$ is symmetric positive definite to find $\mc V^\frac{1}{2}$ and study the similar
\[
\mc V^\frac{1}{2} \mc J^\gamma \mc V^\frac{1}{2} \sim \mc V \mc J^\gamma.
\]

We now aim to adapt the arguments in \cite[Section 3]{ammari.barandun.ea2025Universal} to $\mc V^\frac{1}{2} \mc J^\gamma \mc V^\frac{1}{2}$ to obtain ergodicity and spectral convergence results due to metric transitivity.

We need the following definition. 
\begin{definition}[Truncations]
    We define the $K$-truncation $A_K\in \mathbb{R}^{(2K+1)\times (2K+1)}$ of some operator $A:\lz\to \lz$ to be 
    \[
        (A_K)_{i,j} = A_{i,j} \quad  \text{for } \abs{i}, \abs{j} \leq K.
    \]

    For some sequence $\chi \in \ldz$, we further define the $M$-truncation $\chi_M\in \mathbb{L}_D(2M+1)$ to be 
    \[
        \chi_M^{(i)} = \chi^{(i)} \quad \text{for } \abs{i} \leq M.
    \]
\end{definition}

We now find that the empirical cumulative density of states of the $K$-truncations of $\VJV$ converge to a non-random limit as $K\to \infty$.
\begin{theorem}\label{thm:ergodic_convergence}
    Let $\chi\in \ldz$ be a random \emph{i.i.d.} sequence of blocks sampled, respectively, with probability $p_1, \dots p_D$, and let $\VJV$ be the associated Jacobi operator on $\lz$.  Then, there exists a non-random positive and continuous measure $D(\VJV,d\lambda)$ such that, with probability $1$,
    \[
        D((\VJV)_K,d\lambda) \to  D(\VJV,d\lambda),
    \]
    weakly as $K\to \infty$.
\end{theorem}
\begin{proof}
    This result closely mirrors \cite[Theorem 3.9]{ammari.barandun.ea2025Universal}. It is straightforward to verify that the entries $\VJV$ are again given by homogeneous Markov chains with finite state space, ensuring the metric transitivity. 
\end{proof}

\subsection{Convergence result in the non-reciprocal setting}
We can now employ the fact that the $K$-truncations of $\VJV$ corresponds to finite block disordered systems obtained by $M$-truncating the block sequence $\chi \in \ldz$ (up to some negligible edge effects), to find that the spectral distributions of these finite block disordered systems also converge to the deterministic limit $D(\VJV,d\lambda)$.
\begin{corollary}\label{cor:spectrapconv}
    Let $\chi\in \ldz$ as in \cref{thm:ergodic_convergence}, $\chi_M$ its $M$-truncation and $\mc C^\gamma (\chi_M)$ the generalised gauge capacitance matrix of the associated finite block disordered system consisting of $M$ blocks. Then, with probability $1$,
    \[
        D(\mc C^\gamma (\chi_M),d\lambda) \to  D(\VJV,d\lambda),
    \]
    weakly as $M \to \infty$, for $D(\VJV,d\lambda)$ as in \cref{thm:ergodic_convergence}.
\end{corollary}
\begin{proof}
    For a given $\chi\in \ldz$ any $M$-truncation determines a finite block disordered system consisting of $N$ individual resonators. After potentially relabelling these resonators, we can find a $K>0$ such that 
    \[
        V^\frac{1}{2}J^\gamma V^\frac{1}{2} \approx (\VJV)_K,
    \] where $V$ and $J^\gamma$ are the material matrix and the symmetrised gauge capacitance matrix determined by the $M$-truncation $\chi_M$ and $(\VJV)_K$ is the $K$-truncation of the operator $\VJV$ on $\lz$. By definition, we have $K\to \infty$ as $M\to \infty$, and the approximate equality is to be understood as an equality up to a fixed number of entries at the edge. However, as in \cite{ammari.barandun.ea2025Universal}, these edge effects become negligible as $M,K\to \infty$ ensuring 
    \[
        D(V^\frac{1}{2}J^\gamma V^\frac{1}{2},d\lambda) \to  D(\VJV,d\lambda).
    \]

    The result now follows from the similarities 
    \[
        V^\frac{1}{2}J^\gamma V^\frac{1}{2} \sim VJ^\gamma \sim VC^\gamma = \mc C^\gamma (\chi_M).
    \]
\end{proof}

\section{Density of states and spectral regions}\label{sec:dos}
In this section, we aim to more precisely describe the spectral distribution of non-reciprocal block disordered systems from the properties of their constituent blocks. To that end, we employ the propagation matrix formalism as well as results from Toeplitz theory. 

\subsection{Propagation matrices and bandgaps}
Analogously to \cite{ammari.barandun.ea2025Subwavelength}, our aim in this subsection is to characterise the spectral gaps of $\capmatg$ in terms of the spectral gaps of its constituent blocks. In order to apply the tools developed in the Hermitian case, we study the spectrum of the similar matrix $VJ^\gamma$. 

Recall that in \cref{ssec:transfer and propmat} we defined the symmetrised transfer matrix $\Tilde{T}_j(\lambda)$ taking any solution of $(VJ^\gamma -\lambda I_N)\bm v=0$ forward by one resonator. Using the \emph{total transfer matrix} $\Tilde{T}_{tot} = \prod_{i=1}^N \Tilde{T}_i$, we can characterise the spectrum of $VJ^\gamma$ as follows:
\begin{equation}\label{eq:transfermat_spectrumcharact}
    \lambda \in \sigma(VJ^\gamma) \iff \Tilde{T}_{tot}(\lambda)\underbrace{\begin{pmatrix}
        1 \\ \sqrt{\frac{\zeta(-\gamma_1\ell_1)}{\zeta(\gamma_0\ell_0)}}
    \end{pmatrix}}_{(\veci{v}{1}, \veci{v}{0})^\top} = \underbrace{\nu\begin{pmatrix}
        \sqrt{\frac{\zeta(\gamma_N\ell_N)}{\zeta(-\gamma_{N+1}\ell_{N+1})}}\\
        1
    \end{pmatrix}}_{(\veci{v}{N+1}, \veci{v}{N})^\top}.
\end{equation}
Here, the exact forms of $(\veci{v}{1}, \veci{v}{0})^\top$ and $(\veci{v}{N+1}, \veci{v}{N})^\top$ follow from the first and last row of $(VJ^\gamma -\lambda I_N)\bm v=0$. 

We now aim to restate \cref{eq:transfermat_spectrumcharact} in terms of symmetrised propagation matrices. To do this, we define the \emph{total propagation matrix} $\Tilde{P}_{tot} = \prod_{i=1}^N \Tilde{P}_i$. Using \cref{eq:propmat-transfermat similarity}, we find that $\Tilde{T}_i = Q_i^{-1}\Tilde{P}_iQ_{i-1}$ as well as 
\[
    \Tilde{T}_{tot} = Q_N^{-1}\Tilde{P}_{tot}Q_0.
\]
This allows us to recharacterise \cref{eq:transfermat_spectrumcharact} as 
\begin{equation}\label{eq:propmat_spectrumcharact}
    \lambda \in \sigma(VJ^\gamma) \iff \Tilde{P}_{tot}(\lambda)\begin{pmatrix}
        1 \\ 0
    \end{pmatrix} = \nu\begin{pmatrix}
        1\\
        0
    \end{pmatrix}.
\end{equation} 

We have thus reduced the question of whether $\lambda$ lies in the spectrum of $VJ^\gamma$ and thus $\capmatg$ to a question about the dynamics of the propagation matrix $\Tilde{P}_{tot}(\lambda)$ stemming from the Hermitian system $VJ^\gamma$ and thus having $\det \Tilde{P}_{tot} = \prod_{i=1}^N\det \Tilde{P}_i(\lambda) = 1$. Therefore, it is natural to again employ the tools developed for the Hermitian setting in \cite[Section 4]{ammari.barandun.ea2025Subwavelength}. 

Crucially, because, as in the Hermitian case, the symmetrised propagation matrix $\Tilde{P}_i(\lambda)$ only depends on the properties of the $i$\textsuperscript{th} resonator, we can define the \emph{symmetrised block propagation} $\Tilde{P}_{B_d}(\lambda)$ matrices as follows: 
\begin{equation}\label{eq:symmblockprop}
    \Tilde{P}_{B_d}(\lambda) \coloneqq \prod_{k=1}^{\len(B_d)}\Tilde{P}_{v_k(B_d), \ell_k(B_d), s_k(B_d),\gamma_k(B_d)}(\lambda).
\end{equation}
This allows us to decompose the total propagation matrix $\Tilde{P}_{tot}$ as the product of block propagation matrices:
\begin{equation*}
    \Tilde{P}_{tot}(\lambda) = \prod_{j=1}^M\Tilde{P}_{B_{\oo{j}}}(\lambda).
\end{equation*}

Analogously to the Hermitian case, we say that the frequency $\lambda\in \R$ lies in the \emph{pass band} of block $B_d$ if we have $\abs{\tr \Tilde{P}_{B_d}(\lambda)}\leq2$ and in the \emph{bandgap} if $\abs{\tr \Tilde{P}_{B_d}(\lambda)}>2$. Here, $\tr$ denotes the trace.
Furthermore, since every $\Tilde{P}_{B_d}(\lambda)\in \operatorname{SL}(2,\R)$ we can identify the block propagation matrices with Möbius transformations on the real projective space $\R\mathbb{P}^1$, which in turn is diffeomorphic to the sphere $\R\mathbb{P}^1\simeq \mathbb{S}^1$. In this language, $\lambda$ lying in the bandgap of $\Tilde{P}_{B_d}(\lambda)$ is equivalent to the corresponding Möbius transformation being hyperbolic. Such a transformation then has two distinct fixed points, one source and one sink which we shall denote by $s(\Tilde{P}_{B_d}(\lambda))\in \R\mathbb{P}^1$ and $u(\Tilde{P}_{B_d}(\lambda)) \in \R\mathbb{P}^1$ respectively\footnote{Namely, let $(\xi_1, E_1), (\xi_2, E_2)$ be the eigenvalues and -spaces of $\Tilde{P}_{B_d}(\lambda)$ such that $\abs{\xi_1}<1<\abs{\xi_2}$. After identifying these dimension $1$ eigenspaces with points in $\mathbb{P}\R^1$ we find that $s(\Tilde{P}_{B_d}(\lambda))=E_1$ and $u(\Tilde{P}_{B_d}(\lambda))=E_2$.}.

In this formalism, Theorem $4.1$ in \cite{ammari.barandun.ea2025Subwavelength} holds for $\VJV$ and $\Tilde{P}_{B_{\oo{j}}}(\lambda)$. Therefore, we obtain the following result of Saxon--Hutner type for the non-reciprocal setting.

\begin{figure}
    \centering
    \includegraphics[width=0.95\textwidth]{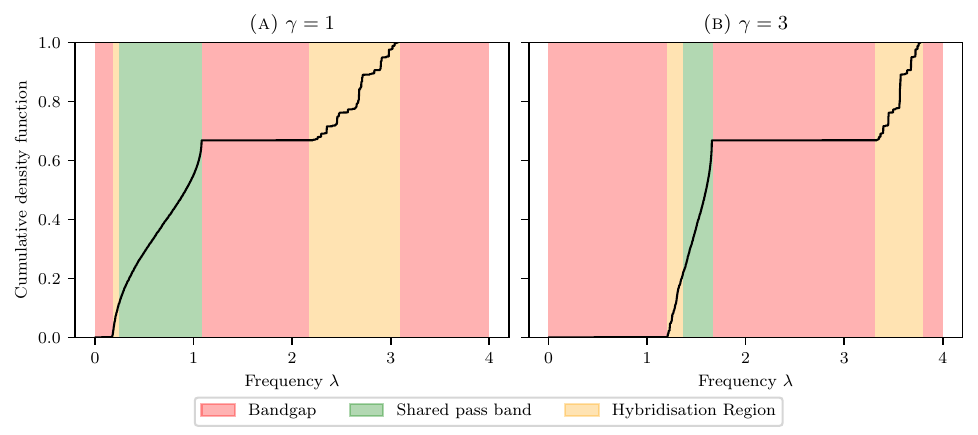}
    \caption{Empirical cumulative density of states for two block disordered systems ($M=1000$, standard blocks as in \cref{ex:standard_blocks}) with $\gamma=1$ and $\gamma=3$, overlaid over the different spectral regions. We can see that the density of states arranges smoothly in the shared pass band, vanishes in the bandgap and behaves fractal like in the hybridisation regions. As $\gamma$ increases, the pass bands of the constituent blocks contract.}
    \label{fig:spectralregions}
\end{figure}
\begin{theorem}\label{thm:saxonhutner}
    Consider an infinite non-reciprocal block disordered system with blocks $B_1, \dots, B_D$, sequence $\chi \in \ldz$ and corresponding symmetrised Jacobi operator $\VJV$. 
    
    Let $\lambda\in \R$ be a frequency lying in the bandgap of all constituent blocks, i.e.
    \[
        \abs{\tr P_{B_d}(\lambda)} > 2 \quad \text{ for all }d=1,\dots, D,
    \]
    and assume further that all sink fixed points $u(\Tilde{P}_{B_d}(\lambda)),\; d=1,\dots,D$ lie in the same connected component of $\R\mathbb{P}^1\setminus\{s(\Tilde{P}_{B_1}(\lambda)), \dots, s(\Tilde{P}_{B_D}(\lambda))\}$\footnote{The monomer and dimers block in \cref{ex:standard_blocks} are specifically chosen such that this source-sink condition is satisfied for all $\lambda$ in the shared bandgap.}.
    
    Then $\lambda$ must also lie in the bandgap of the entire system, \emph{i.e.} $\lambda\notin \sigma(\VJV)$.
\end{theorem}
This result is illustrated in \cref{fig:spectralregions}. We can see that, even for finite systems $\chi\in \ldm$, in the shared bandgaps, the slope of the empirical cumulative density of states $D(\mc C^\gamma, \lambda)$ is in fact zero. 
This observation is justified by the fact that \cref{cor:spectrapconv} ensures, at least in terms of density, the convergence of the finite spectra to the infinite limit. Indeed, this result is made even more precise in \cite{UniformHyperbolicityInPrep}, where the inclusion of the finite spectra into their infinite limit is proven to be sharp, up to some clearly characterisable edge modes.

Furthermore, as in \cite{ammari.barandun.ea2025Universal}, we again observe a tripartite decomposition of the spectral regions. Namely, we say that a frequency $\lambda$ lies in the
\begin{description}
    \item[Shared pass band] if and only if it is in the pass band of all constituent blocks;
    \item[Bandgap] if and only if it is in the bandgap of all constituent blocks;
    \item[Hybridisation region] otherwise.
\end{description}

As in the Hermitian setting, one can observe that while the eigenvalues arrange into smooth bands in the shared pass band, the behaviour in the hybridisation region is much more fractal. In the following section, we will also see that due to the non-reciprocal decay introduced by the imaginary gauge potentials, there is an intricate interplay between skin-like modes localised at the edge and modes localised in the bulk due to disorder.

\section{Competing mechanisms of localisation}\label{sec:competition}
As the systems under consideration are both non-reciprocal and disordered, we observe two distinct mechanisms of localisation. Namely, on the one hand, the non-reciprocity pushes the eigenmodes to be localised at the edge as a result of the imaginary gauge potentials and on the other hand, the disorder causes the eigenmodes to become localised in the bulk. In contrast to the disordered Hermitian or periodic non-reciprocal setting, our systems exhibit both mechanisms simultaneously, leading to an intricate competition between edge and bulk localisation.

\subsection{Competition between localisations at the edge and in the bulk}
To understand this competition between the two different localisations, we investigate the Lyapunov exponent, which, as shown in \cite{ammari.barandun.ea2025Subwavelength}, allows us to predict the decay of an eigenmode associated with a given frequency.  We define the Lyapunov exponent $L$ at the frequency $\lambda$ as
\begin{equation}
    L(\lambda)  \coloneqq \frac{1}{N}\ln \norm{P_{tot}(\lambda)},
\end{equation}
where $\norm{\cdot}$ denotes the \emph{spectral norm} and posit the following characterisation of localisation:
\begin{equation}\label{eq:lyapunovcharact}
    \left\{
    \begin{aligned}
    \text{bulk localisation} \quad & \text{ if }L(\lambda)>0;\\
    \text{delocalised mode} \quad& \text{ if }L(\lambda)=0;\\
    \text{edge localisation} \quad & \text{ if }L(\lambda)<0.
    \end{aligned}
    \right.
\end{equation}
The first two cases mirror the spectral characterisation in the Hermitian setting \cite{ammari.barandun.ea2025Subwavelength}, while the last case $L(\lambda)<0$ can only occur in the non-reciprocal setting where $\det P_{tot}(\lambda)<1$.

Using the symmetrised propagation matrix $\Tilde{P}_{tot}$, the competition between the edge and the bulk localisations becomes apparent.
In fact, from $P_{l_i,s_i,\gamma_i}= e^{-\frac{1}{2}\gamma_il_i}\Tilde{P}_{l_i,s_i,\gamma_i}(\lambda)$, we immediately find that
\begin{equation}
    L(\lambda)  = \underbrace{\frac{1}{N}\ln \norm{\Tilde{P}_{tot}(\lambda)}}_{>0} \underbrace{- \frac{1}{2N}\sum_{i=1}^N\gamma_i\ell_i}_{<0}. 
\end{equation}

Clearly, there is a competition between the positive Lyapunov exponent of $\Tilde{P}_{tot}(\lambda)$, indicating bulk localisation, and the frequency-independent non-reciprocal decay $-\frac{1}{2N}\sum_{i=1}^{N}\gamma_i\ell_i$, indicating localisation at the edge. 

Note that the Lyapunov exponent of the symmetrised propagation matrix $\Tilde{P}_{tot}(\lambda)$ behaves analogously to the Hermitian case. That is, it characterises the spectral regions as follows: 
\begin{description}
    \item[Shared pass band] very small but positive Lyapunov exponent due to the localisation inherent to disordered one-dimensional systems;
    \item[Bandgap] positive Lyapunov exponent, well approximated by the weighted average of the constituent block Lyapunov exponents;
    \item[Hybridisation region] positive Lyapunov exponent, though systematically lower than the weighted average of the constituent block Lyapunov exponents due to the fact that eigenmodes hybridise.
\end{description}

\begin{figure}
    \centering
    \includegraphics[width=0.95\textwidth]{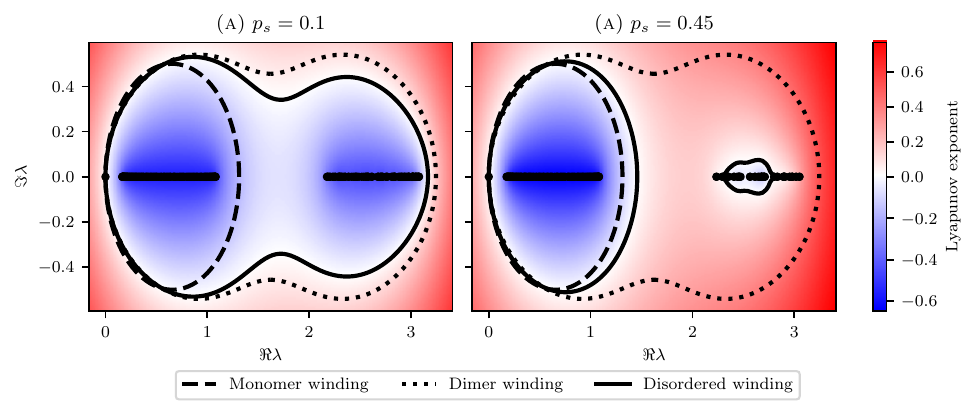}
    \caption{Spectrum, block symbols and total Lyapunov exponents for block disordered systems ($M=100, \gamma=1$, standard blocks) with, respectively, low and high monomer density. We can see that the total Lyapunov exponent exhibits both positive (band localisation) and negative (non-Hermitian skin effect) regions, and the edge contour (\emph{i.e.}, $L(\gamma)=0$, drawn in black) lies between the monomer and dimer block symbols.}
    \label{fig:disordered-winding}
\end{figure}
Since there are no eigenmodes in the bandgap and the bulk localisation in the shared pass band is very weak, the competition between bulk localisation and edge localisation is clearest in the hybridisation regions. In \cref{fig:disordered-winding}, we plot the Lyapunov exponent $L(\lambda)$ for two disordered block systems consisting of monomer and dimer blocks as in \cref{ex:standard_blocks} with $\gamma=1$ for both blocks. As can be seen in \cref{fig:spectralregions}(\textsc{a}), the constructed disordered systems exhibit a hybridisation region around $[2.2,3.1]$, lying in the bandgap of the monomer block and the pass band of the dimer block. 

Because both blocks have the same $\gamma$ and the same total resonator length, the non-reciprocal decay $-\frac{1}{2N}\sum_{i=1}^{N}\gamma_i\ell_i$ remains unchanged as the monomer block probability $p_m$ increases. However, this is not the case for the bulk localisation. In fact, as the monomer block probability $p_m$ increases, any eigenmode in the hybridisation region must pass through more monomer blocks. Because these eigenmodes lie in the bandgap of these blocks, this in turn causes them to become more strongly localised in the bulk. We refer to \cite[Section 4]{ammari.barandun.ea2025Subwavelength} for a more detailed discussion of this type of bandgap localisation.

\begin{figure}
    \centering
    \includegraphics[width=0.95\textwidth]{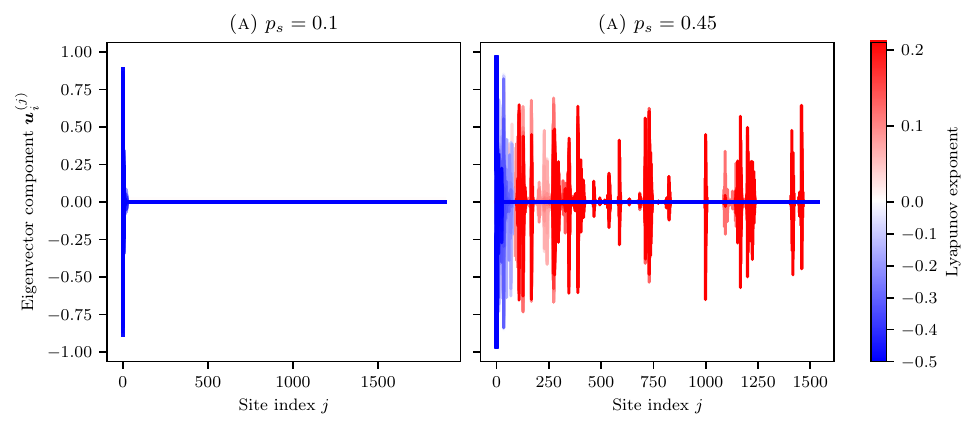}
    \caption{Eigenmodes of the block disordered systems in \cref{fig:disordered-winding}, coloured according to their respective Lyapunov exponents. Negative Lyapunov exponent corresponds to eigenvectors localised at the edge while a positive Lyapunov exponent corresponds to eigenvectors localised in the bulk.}
    \label{fig:eve_competition}
\end{figure}

Consequently, as the monomer block probability $p_m$ increases, the bulk localisation strength increases compared to the edge localisation. This can be observed in \cref{fig:disordered-winding}. For low monomer block probability $p_m=0.1$, all eigenmodes in the hybridisation region have negative Lyapunov exponents, while for higher monomer block probability $p_m=0.45$, we indeed observe some eigenmodes in the hybridisation region with a positive Lyapunov exponent.

This effect is further illustrated in \cref{fig:eve_competition} where for the same systems we plot the eigenmodes, coloured according to their Lyapunov exponents. We can see that the eigenmodes localise as predicted by \cref{eq:lyapunovcharact} and that for low $p_m=0.1$ all eigenmodes are localised at the edge, while there are several eigenmodes localised in the bulk for higher $p_m=0.45$.

This behaviour can alternatively be understood in terms of Toeplitz winding regions. In \cref{fig:disordered-winding}, we determine the Toeplitz winding regions for both the monomer block and the dimer block by imposing quasiperiodic boundary conditions and tracing the resulting bands. Following the theory of Toeplitz matrices (see, for instance, \cite{ammariMathematicalFoundationsNonHermitian2024, ammari.barandun.ea2024Spectra, ammari.barandun.ea2025Generalized}), the eigenmodes within these winding regions must be exponentially localised at the edge. Furthermore, the eigenmodes at the boundary of these winding regions must be delocalised, allowing them to alternatively be characterised by the Lyapunov condition $L(\lambda) = 0$. 

In \cref{fig:disordered-winding}, we draw the contour $L(\lambda) = 0$ for the Lyapunov exponent of the total disordered system and can see that, as in the Toeplitz case, it cleanly divides the complex plane into an exterior region with a positive Lyapunov exponent and an interior region with a negative Lyapunov exponent. 

Furthermore, these disordered winding regions appear to be an interpolation of the winding regions for the monomer and dimer blocks. This is intuitive in light of the fact that we recover the periodic dimer and periodic monomer settings in the extreme cases of $p_m=0$ and $p_m=1$. 
Consequently, as the monomer block density increases, the disordered winding region moves to more closely resemble that of the monomer block, causing it to shrink around the hybridisation region, which in turn causes some of the eigenvalues to lie outside the disordered winding regions and become localised in the bulk.

\subsection{Lyapunov exponent estimate}
\begin{figure}
    \centering
    \includegraphics[width=0.95\textwidth]{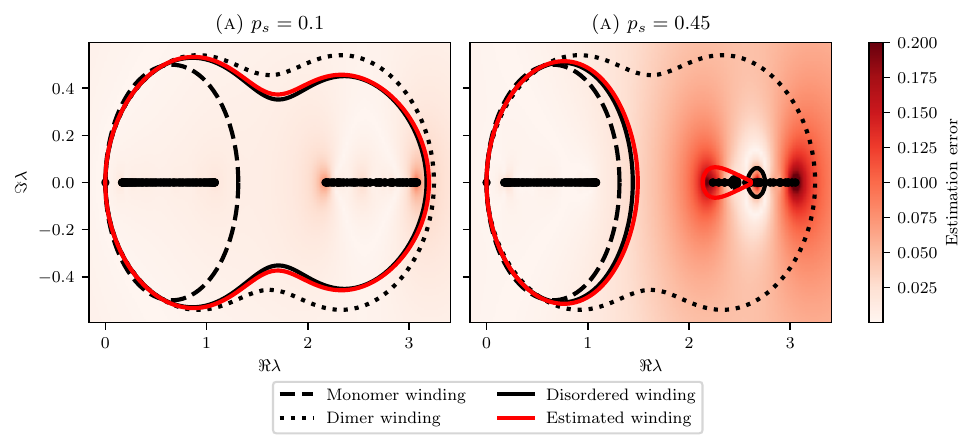}
    \caption{Spectrum, block symbols and Lyapunov exponent estimation error for the same block disordered systems as those in \cref{fig:disordered-winding}. We can see that the estimate is quite accurate away from the eigenvalues and allows for an accurate prediction of the disordered winding region.}
    \label{fig:Lyapunov-estimate}
\end{figure}

In this subsection, we aim to adapt the Lyapunov exponent estimate from \cite[Section 4.3]{ammari.barandun.ea2025Subwavelength} to the non-reciprocal setting. 

Applying the estimate $$\ds \frac{1}{N}\ln \norm{\Tilde{P}_{tot}(\lambda)}\approx \mathbb{E} \bigg[\frac{M}{N} \ln \rho(\Tilde{P}_{B_{\oo{0}}}(\lambda)) \bigg]$$ 
to the symmetrised propagation matrix with $\mathbb{E}$ denoting the expectation, we find the following estimation:
\begin{equation}\label{eq:symmlyapunovestimate}
 \ds  \frac{1}{N}\ln \norm{\Tilde{P}_{tot}(\lambda)} \approx \ds \frac{\ds \sum_{d=1}^Dp_d \ln \rho(\Tilde{P}_{B_d}(\lambda))}{\ds \sum_{d=1}^Dp_d \len(B_d)},
\end{equation}
where $\rho (A)$ denotes the \emph{spectral radius} of the matrix $A$.
Furthermore, due to the law of large numbers, the non-reciprocal decay $\frac{1}{2N}\sum_{i=1}^N\gamma_il_i$ can be very accurately estimated by
\begin{equation}\label{eq:gammaestimate}
    \begin{gathered}
    \frac{1}{2N}\sum_{i=1}^N\gamma_i\ell_i = \frac{1}{2M}\sum_{j=1}^M\frac{M}{N}\sum_{k=1}^{\len (B_{\oo{j}})}\gamma_k(B_{\oo{j}})\ell_k(B_{\oo{j}}) \\\approx \mathbb{E} \bigg[\frac{M}{2N}\sum_{k=1}^{\len (B_{\oo{0}})}\gamma_k(B_{\oo{0}})\ell_k(B_{\oo{0}}) \bigg] = \ds \frac{1}{2}\frac{\ds \sum_{d=1}^Dp_d\sum_{k=1}^{\len (B_d)}\gamma_k(B_d)\ell_k(B_d)}{\ds \sum_{d=1}^Dp_d \len(B_d)},
    \end{gathered}
\end{equation}
where $\gamma_k(B_d)$ and $\ell_k(B_d)$ denote the imaginary gauge potential and the length of the $k$\textsuperscript{th} resonator of block $B_d$, respectively.

Combining \cref{eq:symmlyapunovestimate} and \cref{eq:gammaestimate}, we get the following estimate for the total Lyapunov exponent:
\begin{equation}
    L(\lambda) \approx  \frac{\ds \sum_{d=1}^Dp_d \left(\ln \rho(\Tilde{P}_{B_d}(\lambda)) - \frac{1}{2}\sum_{k=1}^{\len (B_d)}\gamma_k(B_d)\ell_k(B_d)\right)}{\ds \sum_{d=1}^Dp_d \len(B_d)}.
\end{equation}

The agreement of this estimate with the actual Lyapunov exponent $L(\lambda)$ is illustrated in \cref{fig:Lyapunov-estimate}. We can see that, as in the Hermitian setting, this estimate becomes more accurate as $\lambda$ moves further into the bandgaps and struggles, especially around the upper band. 
Nevertheless, it allows for a quite close reconstruction of the disordered winding region, and thus the distinction between the eigenmodes localised in the bulk and those localised at the edge.

\subsection{Disorder acting as insulation}
Finally, in this subsection, we aim to highlight one particular consequence of the competition between bulk and edge localisations. Namely, how the disorder acts as an \emph{insulation} against the non-Hermitian skin effect, preventing edge localisation for small $\gamma$. In contrast, for finite periodic systems consisting of repeated identical unit cells, any small $\gamma>0$ leads to complete localisation of all eigenmodes at the edge.

To illuminate this phenomenon, we consider non-reciprocal block disordered systems consisting of monomers and dimers as in \cref{ex:standard_blocks} with a uniform imaginary gauge potential $\gamma$. This has the advantage that, following \cref{eq:gammaestimate}, the non-reciprocal decay $\frac{1}{2N}\sum_{i=1}^N\gamma_i\ell_i$ is simply approximated by $\frac{2}{3}\gamma$. 

As discussed in \cref{fig:spectralregions}, these blocks lead to systems exhibiting both shared pass bands and hybridisation regions. In \cref{fig:disordered-winding,fig:eve_competition}, we observe bulk localisation in the hybridisation region for $\gamma=1$, but none in the shared pass band. This is because in the shared pass band, the bulk localisation due to disorder is much weaker than the edge localisation induced by $\gamma=1$.

\begin{figure}
    \centering
    \includegraphics[width=0.6\linewidth]{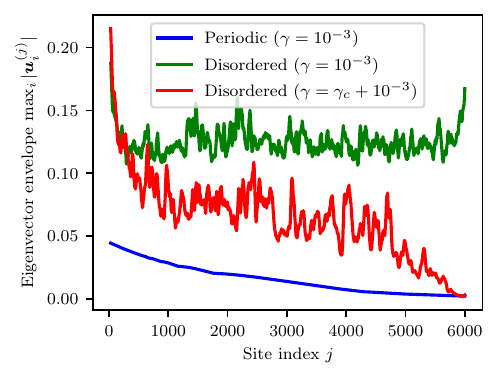}
    \caption{Eigenvector envelope $j\mapsto \max_{\lambda_i\leq\frac{3}{2}} \abs{\bm u_i^{(j)}}$ for both finite periodic ($p_m=0$) and block disordered ($p_m=\frac{1}{2}$) systems ($N=6000$, standard blocks) with varying $\gamma$. While any $\gamma>0$ leads to an exponentially decaying envelope in the periodic case, in the block disordered case, this only holds for $\gamma>\gamma_c$.}
    \label{fig:eigenvector-envelope}
\end{figure}

However, for very small $\gamma$, such a bulk localising effect can even be observed in the shared pass band. In \cref{fig:eigenvector-envelope}, we plot the \emph{eigenvector envelope} over the lower shared pass band\footnote{For small $\gamma>0$, this shared pass band is approximately $[0,1]$. However, since there is a bandgap between this lower band and the upper hybridisation region, it is simpler to just identify the eigenvalues in the lower band by the condition $\lambda_i\leq \frac{3}{2}$.} (\emph{i.e.}, the mapping $j\mapsto \max_{\lambda_i\leq\frac{3}{2}} \abs{\bm u_i^{(j)}}$) for periodic as well as block disordered systems. 

In the periodic case, which is obtained by only choosing dimer blocks $p_m=0$, even for a small $\gamma=10^{-3}$, the eigenvector envelope decreases exponentially with the site index $j$, indicating an edge localisation for \emph{all} modes in the shared pass band. In contrast, for the block disordered case with $p_m=\frac{1}{2}$ and $\gamma=10^{-3}$, the envelope stays of order one as the site index $j$ increases. This is indicative of the fact that the non-reciprocal edge localisation cannot overcome the disorder, leading to a localisation in the bulk.

By numerically calculating the maximal Lyapunov exponent $L(\lambda_i)$ over all eigenvalues in the shared pass band and using the estimate $\frac{1}{2N}\sum_{i=1}^N\gamma_i\ell_i \approx \frac{2}{3}\gamma$, we are able to identify the \emph{critical imaginary gauge potential}, $\gamma_c\approx 10^{-2}$, that is strong enough to cause edge localisation\footnote{This calculation hinges on the fact that $\frac{1}{N}\ln\norm{\Tilde{P}_{tot}(\lambda)}$ is insensitive with respect to $\gamma$, which empirically holds for small $\gamma>0$.}.
In fact, by choosing $\gamma=\gamma_c+10^{-3}$ slightly \emph{above} $\gamma_c$, we can see that the eigenvector envelope now decays to zero as the site index $j$ increases, indicating a localisation at the edge for all eigenvectors in the shared pass band.

\section*{Acknowledgments}
  The work of AU was supported by the Swiss National Science Foundation grant number 200021--200307. 

\section*{Code availability}
The software used to produce the numerical results in this work is openly available at \\ \href{https://doi.org/10.5281/zenodo.15793645}{https://doi.org/10.5281/zenodo.15793645}.

\printbibliography

\appendix

\end{document}